\documentclass[10pt]{article}
\usepackage{macros}
\usepackage{amsthm}
\usepackage{titlesec}

\titleformat*{\section}{\large\bfseries}
\titleformat*{\subsection}{\normalsize\bfseries}
\titleformat*{\subsubsection}{\normalsize\em}

\newtheorem{theorem}{Theorem}[section]

\newtheorem{lemma}[theorem]{Lemma}
\newtheorem{corollary}[theorem]{Corollary}

\theoremstyle{remark}

\newtheorem{definition}[theorem]{Definition}
\newtheorem{remark}[theorem]{Remark}
\newtheorem{example}[theorem]{Example}

\input xy
\xyoption{all}

\title{Dependent choice as a termination principle}
\author{Thomas Powell}
\date{Preprint, \today}
\begin{document}

\maketitle

\begin{abstract}
We introduce a new formulation of the axiom of dependent choice that can be viewed as an abstract termination principle, which generalises the recursive path orderings used to establish termination of rewrite systems. We consider several variants of our termination principle, and relate them to general termination theorems in the literature.
%\keywords{Axiom of choice \and Termination \and Path orderings}
\end{abstract}

%%%%%%%%%%%%%%%%%%%%%%%%%%%%%%%%%%%%%%%%%%%%%%%%%
%%%%%%%%%%%%%%%%%%%%%%%%%%%%%%%%%%%%%%%%%%%%%%%%%
\section{Introduction}
%%%%%%%%%%%%%%%%%%%%%%%%%%%%%%%%%%%%%%%%%%%%%%%%%
%%%%%%%%%%%%%%%%%%%%%%%%%%%%%%%%%%%%%%%%%%%%%%%%%
\label{sec-intro}

Path orderings are a technique for proving that programs terminate. They include the well known multiset \cite{Dershowitz(1982.0)} and Knuth-Bendix \cite{KnuBen(1970.0)} orderings, and are central to the theory of term rewriting.

In order to prove that path orderings are wellfounded, one typically appeals to the axiom of dependent choice in some form. Traditionally, this is via Kruskal's theorem and a clever combinatorial idea of Nash-Williams known as the minimal bad sequence argument, as in \cite{FerZan(1995.0)}, although this can be given a constructive flavour by using bar induction instead \cite{GLar(2001.0)}. 

However, results in the other direction are rare. Dependent choice is often used as a convenient mathematical tool without considering whether or not wellfoundedness of the path order could be established in a weaker theory. Indeed, as shown by Buchholz \cite{Buchholz(1995.0)}, the termination of fixed term rewrite systems via recursive path orderings can in fact even be proven in a weak fragment of Peano arithmetic.

In this note, we single out a general termination principle $\TP$ phrased in a higher-order setting and closely related to open induction \cite{Raoult(1988.0)}, and prove that it is instance-wise equivalent to dependent choice over Peano arithmetic in all finite types. We go on to explore several variants of this principle.

There are several motivating factors behind this work. Firsly, the fact that dependent choice can be viewed as a general wellfoundedness principle akin to recursive path orders is of theoretical interest in its own right, and by making this precise we introduce variants of dependent choice which have deep links to program termination. In this respect our paper is close in spirit to \cite{Mellies(1998.0)}, which establishes a two-way connection between a strong logical principle on one hand and a termination argument on the other.

In addition, by looking at termination on a high level we are able to clarify the relationship between the various generalisations of path orderings one finds in the literature. One variant of our termination principle in particular is based on the notion of a simplification order, and allows us to prove the abstract theorem of Goubault-Larrecq \cite{GLar(2001.0)} as a direct corollary.

Finally, because we present our termination principle as a formal extension of Peano arithmetic in all finite types, which can be given a direct computational interpretation in G\"{o}del's system T via standard proof interpretations, we take a step towards connecting path orderings with higher-order recursion, which is something we briefly mention in Section \ref{sec-conc}. This would extend work begun in \cite{MosPow(2015.0)} and continued in \cite{Powell(2019.1)}, where derivation trees of finitely branching path orders are encoded via terms of system T.

%%%%%%%%%%%%%%%%%%%%%%%%%%%%%%%%%%%%%%%%%%%%%%%%%
%%%%%%%%%%%%%%%%%%%%%%%%%%%%%%%%%%%%%%%%%%%%%%%%%
\section{Preliminaries}
%%%%%%%%%%%%%%%%%%%%%%%%%%%%%%%%%%%%%%%%%%%%%%%%%
%%%%%%%%%%%%%%%%%%%%%%%%%%%%%%%%%%%%%%%%%%%%%%%%%
\label{sec-prelim}

We start by giving a brief overview of the recursive path ordering: Though this is not strictly necessary for the results that follows, it helps motivate them. We then define the main formal systems we will work in for the remainder of the paper.

%%%%%%%%%%%%%%%%%%%%%%%%%%%%%%%%%%%%%%%%%%%%%%%%%
\subsection{The recursive path order}
%%%%%%%%%%%%%%%%%%%%%%%%%%%%%%%%%%%%%%%%%%%%%%%%%
\label{sec-prelim-path}

Let $T$ denote the set of first-order terms build from some finite set of function symbols $F$ and countable set of variables $X$. Suppose that $\succ_F$ is a well-founded order on $F$. For each function symbol $f$ with arity $n$, we assign a \emph{lifting} $\succ_f$, which is a relation on $T^n$ satisfying the property that for any $A\subseteq T$, if $\succ$ is wellfounded on $A$ then $\succ_f$ is wellfounded on $A^n$. The recursive path order $\rpo$ on $T$ is defined recursively as follows: $t=f(t_1,\ldots,t_n)\rpo s$ if either
\begin{enumerate}[(i)]

\item $t_i\rpoeq s$ for some $i=1,\ldots,n$,

\item $s=g(s_1,\ldots,s_m)$ for some $f\succ_F g$, and $t\rpo s_i$ for all $i=1,\ldots,m$,

\item $s=f(s_1,\ldots,s_n)$, $t\rpo s_i$ for all $i=1,\ldots,n$, and $(t_1,\ldots,t_n)\succ_{f}(s_1,\ldots,s_n)$,

\end{enumerate}
In the case where $\succ_f$ is the multiset resp. lexicographic extension of $\succ$, we obtain (variants of) the well-known multiset resp. lexicographic path orderings. The following results are standard in the theory of term rewriting (see \cite{FerZan(1995.0)} for example).

\begin{theorem}\label{res-mpo}The multiset and lexicographic path orders are closed under substitution and contexts, and are well-founded.\end{theorem}

\begin{corollary}\label{res-mpo-trs}Let $\R$ be finite a term rewrite system such that whenever $l\to r$ is a rule of $\R$ then $l\rpo r$. Then the rewrite relation generated by $\R$ is well-founded.\end{corollary}

Recursive path orders allow us to verify that programs defined by a set of rewrite rules are terminating. For example, implementations of many basic primitive recursive functions can be dealt with by the multiset path ordering, while multiply recursive functions such as the Ackermann function are typically reducing under the lexicographic path ordering.

The key feature of path orderings of this kind is that they allow us to prove that \emph{recursively} defined programs terminate. This is the role played by clause (iii) above: Roughly speaking, if $f(t_1,\ldots,t_n)$ only evaluates to a term which contains recursive calls of the form $f(s_1,\ldots,s_n)$ for $(t_1,\ldots,t_n)\succ_f (s_1,\ldots,s_n)$, then rewrite sequences starting from $f(t_1,\ldots,t_n)$ are contained in $\rpo$. Very informally, the reason that $\rpo$ itself is well-founded relies on the fact that whenever we have a sequence of recursive calls
\begin{equation*}f(t_1,\ldots,t_n)\rpo f(s_1,\ldots,s_n)\rpo\ldots\end{equation*} 
where the $t_i, s_j$ are are well-founded with respect to $\rpo$, then that sequence must be finite since $\succ_f$ is a lifting. Such sequences are an example of what we will call \emph{minimal} sequences, in the sense that we assume that all subterms $t_i$ of elements in the sequence are well-founded. 

Minimal sequences, which are a crucial element in most standard proofs that path orderings are well-founded, constitute an idea far more general than the world of path orderings, and the purpose of this paper is to explore it on a much more abstract level. 

%%%%%%%%%%%%%%%%%%%%%%%%%%%%%%%%%%%%%%%%%%%%%%%%%
\subsection{Extensions of Peano arithmetic in all finite types}
%%%%%%%%%%%%%%%%%%%%%%%%%%%%%%%%%%%%%%%%%%%%%%%%%
\label{sec-prelim-PA}

The finite types are defined inductively as follows: $\nat$ and $\bool$ are types, and if $\rho$ and $\tau$ are types then so is the function space $\rho\to\tau$ (which we sometimes write as $\tau^\rho$), the cartesian product $\rho\times\tau$ and finite sequences $\rho^\ast$.  The basic logical system we work in is the theory $\PAomega$ of Peano arithmetic in all finite types, which is just the usual first order theory of Peano arithmetic but now with variables and quantifiers for all types. It also includes the usual combinators for the lambda calculus together with constants for primitive recursion in all types. For a more detailed outline of the kind of theory we have in mind see e.g. \cite{Troelstra(1973.0)} - the precise set up is not important here. We make use of the following notation: For $\alpha,\beta:\nat\to\rho$, $a,b:\rho^\ast$ and $x:\rho$
\begin{itemize}

\item $\initSeg{\alpha}{n}:=\seq{\alpha_0,\ldots,\alpha_{n-1}}$ denotes the initial segment of $\alpha$ of length $n$. The empty sequence (for any type) is denoted $\seq{}$,

\item $|a|:\nat$ denotes the length of $a$,

\item $a\ast x:=\seq{a_0,\ldots,a_{k-1},x}$ the one element extension of $a$ with $x$, and similarly $a\ast b$, $a\ast\beta$ the extension of $a$ with the finite resp. infinite sequence $b$ resp. $\beta$.

\item $a\blacktriangleleft \alpha$ denotes that $\alpha$ is an extension of $a$ i.e. $(\forall i<|a|)(a_i=\alpha_i)$,

\item similarly, $a\blacktriangleleft b$ denotes that $a$ is a (not necessarily strict) prefix of $b$ i.e. $|a|\leq |b|\wedge (\forall i<|a|)(a_i=b_i)$.

\end{itemize}
The axiom schema of dependent choice of type $\rho$ is given by
\begin{equation*}
\DC_\rho \ \colon \ A([])\wedge \forall a^{\rho^\ast}(A(a)\to \exists x^\rho A(a\ast x))\to \exists \alpha^{\nat\to\rho}\forall n A(\initSeg{\alpha}{n})
\end{equation*}
where $A$ is some formula in the language of $\PAomega$. Closely related to dependent choice is bar induction, which in this paper will be given as the following schema of relativised bar induction in all finite types:
\begin{equation*}\BI_\rho \ \colon \ \left\{\begin{aligned}&S(\seq{})\\ &\wedge(\forall\alpha^{0\to\rho}\in S)(\exists n) P(\initSeg{\alpha}{n})\\ &\wedge (\forall a^{\rho^\ast}\in S)((\forall x^\rho)(S(a\ast x)\to P(a\ast x))\to P(a))\end{aligned}  \right\}\to P(\seq{}),\end{equation*} 
where $P$ and $S$ are formulas in the language of $\PAomega$, $a\in S$ is shorthand for $S(a)$ and  $\alpha\in S$ shorthand for $(\forall n^0) S(\initSeg{\alpha}{n})$. We denote by $\PAomega+\DC$ the extension of $\PAomega$ with the axiom schemata $\DC_\rho$ for all finite types, and similarly for $\PAomega+\BI$.

Let $\rhd:\rho\times\rho\to\bool$ be a binary relation, and let 
\begin{equation*}
\TI_\rho[\rhd] \ \colon \ \forall x^\rho (\forall y(x\rhd y\to A(y))\to A(x))\to \forall x A(x)
\end{equation*}
denote the usual principle of transfinite induction over $\rhd$. We denote by $\rlex$ the lexicographic extension of $\rhd$ to infinite sequences of type $\nat\to\rho$ i.e.
\begin{equation*}
\alpha\rlex\beta:\equiv \exists n^\nat (\initSeg{\alpha}{n}=\initSeg{\beta}{n}\wedge \alpha(n)\rhd \beta(n)).
\end{equation*}
Open induction over $\rlex$ is given by the schema
\begin{equation*}
\OI_\rho[\rhd] \ \colon \ \forall \alpha^{\rho^\nat}(\forall \beta(\alpha\rlex \beta\to U(\beta))\to U(\alpha))\to \forall \alpha U(\alpha)
\end{equation*}
where now $U(\alpha)$ is a so-called \emph{open predicate}, which we define to be one of the form $\exists n^\nat B(\initSeg{\alpha}{n})$ for some arbitrary formula $B(s)$ on $\rho^\ast$.
\begin{theorem}
\label{thm-dcs}
The following are provable over $\PAomega$ and hold instance-wise:
\begin{itemize}

\item $\BI_\rho\leftrightarrow \DC_\rho$

\item $\OI_{\rho^\ast\times\bool}[\rhd]\to \DC_\rho$, where $\rhd$ is the relation on $\rho^\ast\times\bool$ defined by $(x,b)\rhd (x',b')$ iff $b=1$ and $b'=0$.

\end{itemize}
Over $\PAomega+\TI_\rho[\rhd]$ we also have
\begin{itemize}

\item $\DC_\rho\to \OI_\rho[\rhd]$

\end{itemize}

\end{theorem}

\begin{proof}
That dependent choice proves bar induction is well-known. The remaining results follow from \cite[Propositions 3.3-3.4]{Berger(2004.0)}.
\end{proof}

%%%%%%%%%%%%%%%%%%%%%%%%%%%%%%%%%%%%%%%%%%%%%%%%%
%%%%%%%%%%%%%%%%%%%%%%%%%%%%%%%%%%%%%%%%%%%%%%%%%
\section{The termination principle $\TP$}
%%%%%%%%%%%%%%%%%%%%%%%%%%%%%%%%%%%%%%%%%%%%%%%%%
%%%%%%%%%%%%%%%%%%%%%%%%%%%%%%%%%%%%%%%%%%%%%%%%%
\label{sec-TP}

Let us now forget about terms and recursive path orderings, which can be encoded in $\PAomega$ in terms of objects and relations of type $\nat$, and replace these with some arbitrary type $\rho$ and relation $\succ$ on $\rho$, which we more generally consider to be a \emph{predicate} on $\rho\times\rho$. We will consider a further relation $\rhd$ on $X$, which intuitively plays the role of the subterm relation, although the only assumption we make here is that $\rhd$ is wellfounded on $\rho$, or in other words, we have access to transfinite induction $\TI[\rhd]$ over $\rhd$.
\begin{definition}
\label{def-WF}We say that $\alpha:\nat\to\rho$ is wellfounded if it satisfies the predicate $\sWF_{\rho}[\succ](\alpha)$ defined by
\begin{equation*}\sWF_{\rho}[\succ](\alpha):\equiv \exists n(\alpha_{n}\nsucc\alpha_{n+1}).\end{equation*}
The relation $\succ$ is wellfounded if $\forall\alpha\sWF_{\rho}[\succ](\alpha)$.
\end{definition}

\begin{definition}
\label{def-MIN}We say that an infinite sequence $\alpha$ is \emph{minimal} with respect to $\rhd$ if all sequences lexicographically less than $\alpha$ are wellfounded with respect to $\succ$. We write this formally via the predicate $\sMIN_{\rho}[\rhd,\succ]$ given by
\begin{equation*}\sMIN_{\rho}[\rhd,\succ](\alpha):\equiv \forall\beta(\alpha\rlex\beta\to \sWF_\rho[\succ](\beta)).\end{equation*}
\end{definition}
Our abstract termination principle is nothing more than a formalisation of the idea briefly discussed in Section \ref{sec-prelim-path}, namely the statement that if all minimal sequences are well-founded, then $\succ$ is well-founded.
\begin{definition}[Termination principle]Given a relation $\rhd$ on $\rho$, we define the schema $\TP_\rho[\rhd]$ as follows:
\begin{equation*}\TP_{\rho}[\rhd]:\equiv \forall\alpha(\sMIN_{\rho}[\rhd,\succ](\alpha)\to\sWF_{\rho}[\succ](\alpha))\to \forall\alpha\sWF_{\rho}[\succ](\alpha).\end{equation*}
where $\succ$ ranges over arbitrary formulas in the language of $\PAomega$.\end{definition}
In the remainder of this paper we will drop the subscripts and/or parameters on $\TP$, $\sMIN(\alpha)$ and $\sWF(\alpha)$ whenever there is no risk of ambiguity.

%%%%%%%%%%%%%%%%%%%%%%%%%%%%%%%%%%%%%%%%%%%%%%%%%
\subsection{$\TP$ as a choice principle}
%%%%%%%%%%%%%%%%%%%%%%%%%%%%%%%%%%%%%%%%%%%%%%%%%
\label{sec-TP-choice}

We now show that $\TP$ has the same strength as the axiom of dependent choice and its variants, and is, moreover, instance-wise equivalent to each of the three choice principles outlined in Section \ref{sec-prelim-PA}. This will follow from Theorem \ref{thm-dcs}, together with the fact that $\TP_{\rho}[\rhd]$ is instance-wise equivalent to $\OI_\rho[\rhd]$, which we now prove.
\begin{lemma}
\label{lem-OI-TP}
$\OI_\rho[\rhd]\to \TP_\rho[\rhd]$ instance-wise over $\PAomega$.
\end{lemma}
\begin{proof}
This is clear, since $\sWF[\succ](\alpha)$ can be expressed as the open predicate 
\begin{equation*}
\exists n B(\initSeg{\alpha}{n})\mbox{ \ \ \ for \ \ \ }B(s):\equiv s_{|s|-2}\nsucc s_{|s|-1}
\end{equation*}
and $\TP[\rhd]$ is nothing more than $\OI[\rhd]$ on the predicate $\sWF[\succ]$.
\end{proof}
\begin{corollary}
\label{cor-CDBI-TP}
$\DC_\rho\to \TP_\rho[\rhd]$ and $\BI_\rho\to \TP_\rho[\rhd]$ instance-wise over $\PAomega+\TI_\rho[\rhd]$.
\end{corollary}

\begin{proof}
Direct from Theorem \ref{thm-dcs} together with the above lemma. 
\end{proof}

Note that the standard proofs of open induction from either dependent choice (via the minimal bad-sequence argument) or bar induction, both lift easily to the termination principle. Proofs of this kind can be found in e.g. \cite{Berger(2004.0),Coquand(1991.0)}, and we give these explicitly in Appendix \ref{sec-app} because they also model important patterns often encountered in the term-rewriting literature, as we will indicate later.

\begin{lemma}
\label{lem-TP-OI}
$\TP_\rho[\rhd^\ast]\to \OI_\rho[\rhd]$ for suitable $\rhd^\ast$, instance-wise over $\PAomega$.
\end{lemma}

\begin{proof}
To derive $\OI[\rhd]$ on the open predicate $U(\alpha):\equiv\exists n^\nat B(\initSeg{\alpha}{n})$, define $\succ$ and $\rhd^\ast$ on $\rho^\ast$ by
\begin{equation*}\begin{aligned}a\succ b&:\equiv (|b|=|a|+1)\wedge (a\blacktriangleleft b)\wedge (\forall c\blacktriangleleft b)\neg B(c)\\
a\rhd^\ast b&:\equiv(|b|\geq |a|)\wedge (\exists i<|a|)(\initSeg{a}{i}=\initSeg{b}{i}\wedge a_{i}\rhd b_i).\end{aligned}\end{equation*}
Assuming the premise of $\OI[\rhd]$ we prove the premise of $\TP[\rhd^\ast]$. Take some minimal $\gamma\in \nat\to \rho^\ast $, which satisfies
\begin{equation*}\label{eqn-TPpprem}(\ast) \ \ \ (\forall \delta)(\gamma\rlex^\ast\delta\to \exists n(\delta_n\nsucc \delta_{n+1})).\end{equation*}
We want to prove $\exists n(\gamma_n\nsucc\gamma_{n+1})$. We can assume w.l.o.g. that $\forall n(|\gamma_{n+1}|=|\gamma_n|+1\wedge\gamma_n\blacktriangleleft\gamma_{n+1})$, else if this were false then by definition there would be some $n$ with $\gamma_n\nsucc\gamma_{n+1}$. Let $N:=|\gamma_0|$ and define the diagonal sequence $\tilde\gamma\in \rho^\nat$ by
\begin{equation*}\tilde\gamma_n:=\begin{cases}(\gamma_0)_n & \mbox{if $n<N$}\\
(\gamma_{m+1})_{N+m} & \mbox{if $n=N+m$},\end{cases}\end{equation*}
which is well-defined since $|\gamma_{m}|=N+m$. Now suppose that $\beta\in \rho^\nat$ is such that $\tilde\gamma\rlex\beta,$ and define $\delta\in (\rho^\ast)^\nat$ by $\delta_n:=\initSeg{\beta}{N+n}$ (so in particular $|\gamma_n|=N+n=|\delta_n|$ for all $n$). Then we must have $\gamma\rlex^\ast\delta$. 

To see this, recall that $\tilde\gamma\rlex\beta$ means there exists some $m$ with  $\initSeg{\tilde\gamma}{m}=\initSeg{\beta}{m}$ and $\tilde\gamma_m\rhd\beta_m$. Then either $m<N$, in which case we have $\initSeg{\gamma_0}{m}=\initSeg{\tilde\gamma}{m}=\initSeg{\beta}{m}=\initSeg{\delta_0}{m}$ and $(\gamma_0)_m\rhd (\delta_0)_m$ and so $\gamma_0\rhd^\ast\delta_0$. Or $m=N+k$ then since $\initSeg{\tilde\gamma}{N+k}=\initSeg{\beta}{N+k}$ and $\tilde\gamma_{N+k}\rhd\beta_{N+k}$ it follows that $\gamma_n=\delta_n$ for all $n\leq k$, $\initSeg{\gamma_{k+1}}{N+k}=\initSeg{\delta_{k+1}}{N+k}$ and $(\gamma_{k+1})_{N+k}\rhd (\delta_{k+1})_{N+k}$ and hence $\gamma_{k+1}\rhd^\ast\delta_{k+1}$.

But if $\gamma\rlex^\ast\delta$ then by ($\ast$) there is some $n$ with $\delta_n\nsucc\delta_{n+1}$. But since $\delta_n=\initSeg{\beta}{N+n}\blacktriangleleft\initSeg{\beta}{N+n+1}=\delta_{n+1}$ this means that $(\exists c\blacktriangleleft\initSeg{\beta}{N+n+1}) B(c)$, or in other words, $B(\initSeg{\beta}{k})$ must hold for some $k\leq N+n+1$, from which $U(\beta)$ follows by definition. Therefore we have shown that $\forall \beta(\tilde\gamma\rlex\beta\to U(\beta))$, and hence by the premise of $\OI[\rhd]$ we obtain $U(\tilde\gamma)$.

But this means that there is some $n$ such that $B(\initSeg{\tilde\gamma}{n})$ holds, and since $\initSeg{\tilde\gamma}{n}\blacktriangleleft \gamma_{n\mathop{\dot -}N}\blacktriangleleft\gamma_{(n\mathop{\dot -}N)+1}$ (where $\mathop{\dot -}$ denotes cut-off subtraction) it follows that $(\exists c\blacktriangleleft \gamma_{(n\mathop{\dot -}N)+1})B(c)$, and therefore must have $\gamma_{n\mathop{\dot -}N}\nsucc \gamma_{(n\mathop{\dot -}N)+1}$. Therefore we have shown in all cases that $\exists n(\gamma_n\nsucc \gamma_{n+1})$ whenever $\gamma$ is minimal with respect to $\rhd^\ast$. This establishes the premise of $\TP_{\rhd^\ast,\succ}$, and so it follows that $(\exists n)(\gamma_n\nsucc\gamma_{n+1})$ holds for \emph{arbitrary} $\gamma$.

Now take some arbitrary $\alpha\in \rho^\nat$ and define $\gamma$ by $\gamma_n=\initSeg{\alpha}{n}$. Since there exists some $n$ with $\gamma_n\nsucc\gamma_{n+1}$ it follows that $(\exists c\blacktriangleleft\initSeg{\alpha}{n+1})B(c)$, in other words $B(\initSeg{\alpha}{k})$ holds for $k\leq n+1$, and thus $U(\alpha)$ holds. Therefore using $\TP[\rhd^\ast]$ we have proved $\OI_{\rhd}[U]$.
\end{proof}

%%%%%%%%%%%%%%%%%%%%%%%%%%%%%%%%%%%%%%%%%%%%%%%%%
%%%%%%%%%%%%%%%%%%%%%%%%%%%%%%%%%%%%%%%%%%%%%%%%%
\section{An equivalent formulation of $\TP$ for well-founded elements}
%%%%%%%%%%%%%%%%%%%%%%%%%%%%%%%%%%%%%%%%%%%%%%%%%
%%%%%%%%%%%%%%%%%%%%%%%%%%%%%%%%%%%%%%%%%%%%%%%%%
\label{sec-equivalent}

So far, our termination principle is essentially a modification of open induction, for open predicates which are restricted to two consecutive elements. In this section, we reformulate $\TP$ so that it more closely resembles a genuine termination argument, in the sense that it deals with well-founded \emph{elements} of $\rho$ rather than sequences. We will then apply this in the next section to provide an abstract termination principle for a generalisation of simplification orders.

\begin{definition}
\label{defn-eWF}
We say that $x:\rho$ is well-founded if it satisfies the predicate 
\begin{equation*}
\eWF_\rho[\succ](x):\equiv \forall\alpha (x\blacktriangleleft\alpha\to\sWF_\rho[\succ](\alpha))
\end{equation*}
\end{definition}

\begin{definition}
\label{defn-eMIN}
We define
\begin{equation*}
%\begin{aligned}
\eMIN_\rho[\rhd,\succ](\alpha):\equiv \forall n,y^\rho(\alpha_{n-1}\succ y\wedge\alpha_n\rhd y\to\eWF_\rho[\succ](y)))
%\end{aligned}
\end{equation*} 
where for $n=0$ the condition $\alpha_{n-1}\succ y$ vanishes.
\end{definition}
\begin{definition}
The termination principle $\eTP_\rho[\rhd]$ is defined as
\begin{equation*}\eTP_\rho[\rhd]:\equiv \forall\alpha(\eMIN_\rho[\rhd,\succ](\alpha)\to \exists n(\alpha_n\nsucc \alpha_{n+1}))\to \forall x\eWF_\rho[\succ](x).\end{equation*}
where $\succ$ ranges over arbitrary formulas in the language of $\PAomega$.\end{definition}

\begin{lemma}\label{res-TP-eTP}$\TP_\rho[\rhd]\leftrightarrow\eTP_\rho[\rhd]$ instance-wise over $\PAomega$.\end{lemma}

\begin{proof}We clearly have $\forall\alpha\sWF(\alpha)\leftrightarrow \forall x\eWF(x)$ and so the result follows if we can show that the premise of $\TP$ is equivalent to that of $\eTP$. 

In one direction, assume that $\forall\alpha(\sMIN(\alpha)\to\sWF(\alpha))$ and $\eMIN(\alpha)$ holds for some fixed $\alpha$. Take some $\beta\llex\alpha$ with $\initSeg{\alpha}{n}=\initSeg{\beta}{n}\wedge\alpha_n\rhd\beta_n$. Either $\alpha_{n-1}=\beta_{n-1}\nsucc \beta_n$ and so $\sWF(\beta)$, or $\alpha_{n-1}\succ\beta_n$ and so by $\eMIN(\alpha)$ we have $\eWF(\beta_n)$ and hence $\sWF(\beta)$  (note that for $n=0$, $\sWF(\beta)$ follows directly from $\alpha_0\rhd\beta_0$, since in this case the requirement $\alpha_{n-1}\succ \beta_n$ vanishes). This establishes $\sMIN(\alpha)$ and therefore $\sWF(\alpha)$. 

For the other direction, assume that $\forall\alpha(\eMIN(\alpha)\to\sWF(\alpha))$ and $\sMIN(\alpha)$ holds for some fixed $\alpha$, and suppose for contradiction that $\neg \sWF(\alpha)$. Take some $n$ and $y$ such that $\alpha_{n-1}\succ y$ and $\alpha_n\rhd y$. Then in particular, for any $y\blacktriangleleft \beta$ we have $\initSeg{\alpha}{n}\ast \beta\llex \alpha$ and therefore $\sWF(\initSeg{\alpha}{n}\ast\beta)$ by $\sMIN(\alpha)$. But since $\alpha_0\succ\ldots\succ\alpha_{n-1}\succ y$, this means that $\sWF(\beta)$, and thus we have established $\eMIN(\alpha)$ and therefore $\sWF(\alpha)$, a contradiction.
\end{proof} 

%%%%%%%%%%%%%%%%%%%%%%%%%%%%%%%%%%%%%%%%%%%%%%%%%
%%%%%%%%%%%%%%%%%%%%%%%%%%%%%%%%%%%%%%%%%%%%%%%%%
\section{Simplification orders}
%%%%%%%%%%%%%%%%%%%%%%%%%%%%%%%%%%%%%%%%%%%%%%%%%
%%%%%%%%%%%%%%%%%%%%%%%%%%%%%%%%%%%%%%%%%%%%%%%%%
\label{sec-simplification}

We now present our final variation of $\TP$, which can be directly related to abstract termination principles as they appear in the term-rewriting literature. The key to this is to introduce an additional property in all finite types which generalises a feature possessed by the majority of well-known path orders in term rewriting: namely that $x\rhd y$ (or more generally $(\exists u)(x\rhd u\succeq y)$) implies $x\succ y$. Recall that in this setting, $\lhd$ plays the role of the subterm relation, and orders which have the aforementioned property are known as \emph{simplification} orders.

Simplification orders can be characterised by a auxiliary relation $\succ_0$ which essentially defines $x\succ y$ in the case that $(\exists u)(x\rhd u\succeq y)$ is not true. In the case of terms in \cite{FerZan(1995.0)}, this splitting up of $\succ$ is called a \emph{decomposition}, and so we use the same terminology here, although of course for us our basic objects are not terms but elements of some arbitrary type $\rho$. 
\begin{definition}\label{defn-decomp}A predicate $\succ_0$ on $\rho\times\rho$ is called a decomposition of $\succ$ if it satisfies the following two properties:
\begin{enumerate}[(a)]

\item\label{item-decomp-a} $x\succ y\to \exists u(x\rhd u\succeq y)\vee x\succ_0 y$;

\item\label{item-decomp-b} $x\succ_0y\to \forall u(y\rhd u\to x\succ u)$.

\end{enumerate}
where $\succeq$ denotes the predicate $x\succ y\vee x=y$. Note that if $x\succ y\to \forall u(y\rhd u\to x\succ u)$ then $\succ$ is a decomposition of itself, although naturally we are interested in cases where $\succ_0$ is a not the same as $\succ$.\end{definition}

\begin{example}
For the recursive path order discussed in Section \ref{sec-prelim-path}, we would define $t=f(t_1,\ldots,t_n)\succ_0 s$ iff
\begin{enumerate}[(i)]

\item $s=g(s_1,\ldots,s_m)$ for some $f\succ_F g$, and $t\rpo s_i$ for all $i=1,\ldots,m$,

\item $s=f(s_1,\ldots,s_n)$, $t\rpo s_i$ for all $i=1,\ldots,n$, and $(t_1,\ldots,t_n)\succ_{f}(s_1,\ldots,s_n)$.

\end{enumerate}
Then $\succ_0$ is clearly a decomposition of $\rpo$ with respect to the immediate subterm relation $\rhd$.
\end{example}

The notion of a decomposition is extremely useful, as it enables us to restrict our attention to wellfoundedness of minimal sequences under the auxiliary relation $\succ_0$, which in practise is usually chosen to be something obviously wellfounded. 

\begin{definition}
Define the predicate $A_\rho[\rhd,\succ](x)$ on $\rho$ by 
\begin{equation*}
A_\rho[\rhd,\succ](x):\equiv (\forall y\lhd x)\eWF[\succ](y)
\end{equation*} 
and define
\begin{equation*}\eWF_{A}[\rhd,\succ,\succ_0](x):\equiv (\forall\alpha\in A^\NN)(x\blacktriangleleft \alpha\to \exists n(\alpha_n\nsucc_0\alpha_{n+1})).\end{equation*}
where $\alpha\in A$ is shorthand for $\forall n A(\alpha_n)$.
\end{definition}

\begin{definition}
The termination principle $\sTP_\rho[\rhd]$ is defined as
\begin{equation*}
\sTP	_\rho[\rhd]:\equiv \forall x\eWF_{A}[\rhd,\succ,\succ_0](x)\to \forall x\eWF[\succ](x)
\end{equation*}
where $\succ$ and $\succ_0$ range over arbitrary formulas in the language of $\PAomega$.
\end{definition}

\begin{theorem}
\label{res-TP-simp}
If $\succ_0$ is a decomposition of $\succ$, then $\eTP[\rhd]\to\sTP[\rhd]$ instance-wise over $\PAomega$. If, in addition, $x\succ_0y\to x\succ y$ then the implication holds in the other direction.\end{theorem}

\begin{proof}For one direction suppose that $\eTP$ and $\forall x\eWF_{A}(x)$ hold. We fix some $\alpha$ and prove $\eMIN(\alpha)\to\sWF(\alpha)$. Suppose for contradiction that $\neg\sWF(\alpha)\wedge\eMIN(\alpha)$ is true. Our first step is to show that $\neg\sWF_{\succ_0}(\alpha)$. Suppose for contradiction that $\alpha_n\nsucc_0\alpha_{n+1}$ for some $n$, and w.l.o.g. take this $n$ to be minimal. Then since we must have $\alpha_n\succ\alpha_{n+1}$ (by $\neg\sWF(\alpha)$), by property (\ref{item-decomp-a}) it can only be that $\alpha_n\rhd u\succeq\alpha_{n+1}$ for some $u$. But by minimality of $n$ we have $\alpha_{n-1}\succ_0\alpha_n$ and hence by property (\ref{item-decomp-b}) of Definition \ref{defn-decomp} we have $\alpha_{n-1}\succ u$. But since both $\alpha_{n-1}\succ u$ and $\alpha_n\rhd u$ it follows from $\eMIN(\alpha)$ that $\eWF(u)$ (note that for $m=0$ the prerequisite $\alpha_{n-1}\succ u$ is redundant), and since $u\succeq\alpha_{n+1}$ this implies that $\eWF(\alpha_{n+1})$ and hence $\sWF(\alpha)$, contradicting $\neg\sWF(\alpha)$.

So we have $\neg\sWF_{\succ_0}(\alpha)$. Now, it follows from $\eMIN(\alpha)$ that for any $n,y$ we have $\alpha_{n-1}\succ y\wedge\alpha_n\rhd y\to\eWF(y)$. But since by $\neg\sWF_{\succ_0}(\alpha)$ we must have $\alpha_{n-1}\succ_0\alpha_n$, and therefore by property (\ref{item-decomp-b}), $\alpha_n\rhd y$ automatically implies $\alpha_{n-1}\succ y$, and so in summary we have shown $(\forall n,y)(\alpha_n\rhd y\to\eWF(y))$, or in other words $\alpha\in A^\NN$. But then $\neg\sWF_{\succ_0}(\alpha)$ contradicts $\eWF_{A}(\alpha_0)$ and thus also our assumption that $\forall x\eWF_{A,\succ_0}(x)$. So $\neg\sWF(\alpha)\wedge\eMIN(\alpha)$ must be false, and since $\alpha$ was arbitrary we have proven the premise of $\eTP$, from which we can infer $\forall x\eWF(x)$.

For the other direction, given our additional assumption $x\succ_0 y\to x\succ y$, suppose that $\forall x\eWF_{A,>_0}(x)\to \forall x\eWF(x)$ and the premise of $\eTP$ hold. Let's take some $x\blacktriangleleft\alpha$ with $\alpha\in A^\NN$. Then it is clear that such an $\alpha$ must satisfy $\sMIN(\alpha)$: Given $n$ and $y$ with $\alpha_{n-1}\succ y$ and $\alpha_m\rhd y$, then by $\alpha_n\in A$ we clearly have $\eWF(y)$. Therefore by the premise of $\eTP$ we have $\sWF(\alpha)$ i.e. $\alpha_n\nsucc\alpha_{n+1}$ for some $n$, and by our additional assumption this implies $\alpha_n\nsucc_0\alpha_{n+1}$ and hence $\sWF_{\succ_0}(\alpha)$. Since $x$ and $\alpha$ were arbitrary we have proved $\forall x\eWF_{A,>_0}(x)$ from which we can infer $\forall x\eWF(x)$, and this establishes $\eTP$.\end{proof}

%%%%%%%%%%%%%%%%%%%%%%%%%%%%%%%%%%%%%%%%%%%%%%%%%
%%%%%%%%%%%%%%%%%%%%%%%%%%%%%%%%%%%%%%%%%%%%%%%%%
\section{A connection with abstract path orderings}
%%%%%%%%%%%%%%%%%%%%%%%%%%%%%%%%%%%%%%%%%%%%%%%%%
%%%%%%%%%%%%%%%%%%%%%%%%%%%%%%%%%%%%%%%%%%%%%%%%%
\label{sec-connection}

Our final termination principle $\sTP$ follows instance-wise from dependent choice and conversely, modulo a small additional assumption, the two are actually equivalent. We now show how $\sTP$ can very much be viewed as a genuine termination principle by showing that it is closely related to the abstract termination theorem of Goubault-Larrecq in \cite[Theorem 1]{GLar(2001.0)}, when the latter is formulated in a typed setting.
\begin{corollary}[Goubault-Larrecq \cite{GLar(2001.0)}]Let $\succ$, $\rhd$ and $\gg$ be three binary relations on $\rho$ such that $x\succ y$ implies that either
\begin{enumerate}[(i)]

\item\label{item-GLi} $x\rhd u\succeq y$ for some $u$, or

\item\label{item-GLii} $x\gg y$ and $(\forall u)(y\rhd u\to x\succ u)$.

\suspend{enumerate}

and furthermore

\resume{enumerate}[{[(i)]}]

\item\label{item-GLiii} $\rhd$ is well-founded, and

\item\label{item-GLiv} for every $x\in X$, if for every $y\lhd x$ we have $\eWF(y)$ then $x$ is accessible in $\gg_{A}$ where $A:\equiv\{x\in X\; | \; (\forall y\lhd x)\eWF(y)\}$.

\end{enumerate}
Then $\forall x\eWF(x)$. \end{corollary}

\begin{remark}Note that technically, condition (\ref{item-GLiv}) above is actually the alternative condition (v) in \cite{GLar(2001.0)}.\end{remark}

\begin{proof}The first assumption that $x\succ y$ implies either (\ref{item-GLi}) or (\ref{item-GLii}) shows that the binary predicate $\succ_0$ given by 
\begin{equation*}x\succ_0y:\equiv x\gg y\wedge (\forall u)(y\rhd u\to x\succ u)\end{equation*}
is a decomposition of $\succ$. The wellfoundedness condition (\ref{item-GLiii}) corresponds to $\TI[\rhd]$, while (\ref{item-GLiv}) is equivalent to the statement $\forall x\eWF_{A}[\gg](x)$. But since $x\succ_0y\to x\gg y$ this in turn implies $\forall x\eWF_{A}[\succ_0](x)$, and therefore by Theorem \ref{res-TP-simp} we have $\forall x\eWF(x)$.\end{proof}

The original proof in \cite{GLar(2001.0)} uses a variant of bar induction. If we were to take the bar inductive proof of $\TP$ (given explicitly in Appendix \ref{sec-app}) and adapt via the proofs of Theorems \ref{res-TP-eTP} and \ref{res-TP-simp} to a bar inductive proof of $\sTP$, it would be closely related to that of \cite{GLar(2001.0)}. 

In addition, \cite{GLar(2001.0)} shows that wellfoundedness of many of the usual path orders, including Fereirra-Zantema's wellfoundedness proof for term orderings \cite[Theorem 4]{FerZan(1995.0)}, follow as a corollary of the above result, and so in turn must also be subsumed by our abstract termination principle. Moreover, were we to adapt the proof of $\TP$ via dependent choice (also given in Appendix \ref{sec-app}) to prove Theorem 4 of \cite{FerZan(1995.0)}, we would end up with a very similar proof based on a minimal bad-sequence construction.

It is interesting to note that although proofs of termination via open induction are much less common, they have been considered from the perspective of formalisation \cite{Sternagel(2017.0)}, where direct inductive argument is much easier to work with than proof which is reliant on classical logic.

All of this demonstrates that $\TP$ and its variants are not only abstract termination theorems in the sense that they subsume well known termination results in the literature, but also the proofs of $\TP$ via $\DC$ and $\BI$ can be seen as abstract representations of common proof techniques seen in the theory of term rewriting. 

To summarise, we have the following chain of termination principles, starting at the most general:
\begin{equation*}\begin{aligned}\DC\leftrightarrow\TP\leftrightarrow\eTP \to \; &\sTP\\ \Rightarrow \; &\mbox{Goubault-Larrecq \cite[Theorem 1]{GLar(2001.0)}}\\ \Rightarrow_{\scriptsize} \; &\mbox{Fereirra-Zantema \cite[Theorem 4]{FerZan(1995.0)}}\\
\Rightarrow \; & \mbox{multiset, lexicographic path orders etc.}\end{aligned}\end{equation*}
where $\Rightarrow$ indicates a mathematical implication which is not explicitly formalised in $\PAomega$. Note that from Fereirra-Zantema onwards, termination theorems deal specifically with terms over some signature, and are thus can be encoded using just $\TP_\nat$ of base type. 

%%%%%%%%%%%%%%%%%%%%%%%%%%%%%%%%%%%%%%%%%%%%%%%%%
%%%%%%%%%%%%%%%%%%%%%%%%%%%%%%%%%%%%%%%%%%%%%%%%%
\section{Concluding remarks}
%%%%%%%%%%%%%%%%%%%%%%%%%%%%%%%%%%%%%%%%%%%%%%%%%
%%%%%%%%%%%%%%%%%%%%%%%%%%%%%%%%%%%%%%%%%%%%%%%%%
\label{sec-conc}

It is hoped that this short note provides some insight into proof theoretic aspects of termination arguments commonly found in term rewriting and related areas, in particular their relation to choice principles.

An interesting next step would be to consider forms of higher-order recursion which constitute natural computational counterparts to our termination principles. For the axiom of open induction, a corresponding recursor called \emph{open recursion} has been considered by Berger \cite{Berger(2004.0)} and shown to give a direct realizability interpretation to $\OI[\rhd]$. 

It is easy to see that essentially the same form of recursion would give a computational interpretation to $\TP[\rhd]$: For a functional $f:\rho^\nat\to (\nat\to \rho\to \rho^\nat\to\nat)\to\nat$ satisfying the modified realizability interpretation of the premise of $\TP$ i.e.
\begin{equation*}
\begin{aligned}
\forall\alpha,\phi(&\forall n,y,\beta(\alpha_n\rhd y\to (\initSeg{\alpha}{n}\ast y\ast \beta)_{\phi ny\beta}\nsucc (\initSeg{\alpha}{n}\ast y\ast \beta)_{\phi ny\beta+1})\\
&\to \alpha_{f\alpha \phi}\nsucc \alpha_{f\alpha\phi+1})
\end{aligned}
\end{equation*}
we would have
\begin{equation*}
\forall\alpha(\alpha_{\Phi f\alpha}\nsucc \alpha_{\Phi f\alpha+1})
\end{equation*}
where $\Phi$ is the functional recursively defined by
\begin{equation*}
\Phi f\alpha=f\alpha(\lambda n,y,\beta\; . \; \Phi(\initSeg{\alpha}{n}\ast y\ast \beta)\mbox{ if $\alpha_n\rhd y$ else $0$}).
\end{equation*}
Giving a computational interpretation to $\eTP$ and $\sTP$, on the other hand, would be more difficult, since the proofs of these principles from $\TP$ seem to use classical logic in an essential way. One option would be to either use realizability together with the $A$-translation, or to consider instead the functional interpretation, which has been used to analyse the combination of open induction and classical logic in \cite{Powell(2018.0)}

Particularly intriguing would be to see whether the equivalences proven here give rise to new interdefinability results between forms of recursion as in \cite{Powell(2014.0)}. However, we leave such questions to future work.

\appendix

\section{Direct proof of Corollary \ref{cor-CDBI-TP}}
\label{sec-app}

$\DC_\rho\to \TP_\rho[\rhd]$: This uses the famous minimal bad-sequence construction. Let us call a sequence $\alpha$ \emph{bad} if $\neg\sWF[\succ](\alpha)$ holds: in other words, $\alpha$ is an infinite $\succ$-descending chain. Suppose that the premise of $\TP[\rhd]$ holds, and that for contradiction there exists at least one bad sequence. Using $\DC_\rho$ together with $\TI[\rhd]$, construct a \emph{minimal} sequence $\alpha$ as follows:
\begin{quote}Assuming we have already constructed $\seq{\alpha_0,\ldots,\alpha_{n-1}}$, choose $\alpha_n$ in such a way that $\seq{\alpha_0,\ldots,\alpha_{n-1},\alpha_n}$ extends to a bad sequence, but $\seq{\alpha_0,\ldots,\alpha_{n-1},x}$ does not for any $x\lhd\alpha_n$.\end{quote} 
For the empty sequence in the first step this is guaranteed by the initial assumption that at least one bad sequence exists. It is easy to see that $\alpha$ must satisfy $\sMIN(\alpha)$. However, $\alpha$ itself must also be bad: if on the contrary we would have $\alpha_n\nsucc\alpha_{n+1}$ for some $n$, then $\seq{\alpha_0,\ldots,\alpha_{n+1}}$ could not extend to a bad sequence, a contradiction.\\

\noindent $\BI_\rho\to \TP_\rho[\rhd]$: Define
\begin{equation*}\begin{aligned}S(a)&:\equiv (\forall n<|a|,\beta^{\rho^\nat})(\initSeg{a}{n}\blacktriangleleft\beta\wedge a_n\rhd\beta_n\to\sWF(\beta))\\
P(a)&:\equiv \forall\alpha(a\blacktriangleleft\alpha\to\sWF(\alpha)).\end{aligned}\end{equation*}
From the premise of $\TP$ we derive the three premises of $\BI$ w.r.t $P$ and $S$. Note that $S(\seq{})$ is trivially true, and if $\alpha\in S$ then this is completely equivalent to saying that $\sMIN(\alpha)$ holds, and hence $\alpha_n\nsucc\alpha_{n+1}$ for some $n$ and thus $P(\initSeg{\alpha}{n+2})$ holds.

For the third premise, take some $a\in S$ and assume that $(\forall x)(S(a\ast x)\to P(a\ast x))$. We establish $(\forall x)P(a\ast x)$ via a side induction on $\rhd$, from which we trivially obtain $P(a)$ since for any $\alpha$ with $a\blacktriangleleft\alpha$ we have $a\ast\alpha_{|a|}\blacktriangleleft\alpha$ and therefore $\sWF(\alpha)$ follows from $P(a\ast\alpha_{|a|})$.

Suppose that $(\forall y\lhd x) P(a\ast y)$ holds. Then to prove $P(a\ast x)$ it suffices to prove $S(a\ast x)$. Since we already have $a\in S$, it suffices to check the last point of the sequence i.e.
\begin{equation*}(\forall\beta)(a\blacktriangleleft\beta\wedge x\rhd\beta_{|a|}\to\sWF(\beta)).\end{equation*}
But this follows from the side induction hypothesis, setting $y:=\beta_{|a|}$, which completes the side induction. Therefore, we can now apply bar induction to obtain $P(\seq{})$ which is just $(\forall\alpha)\sWF(\alpha)$.

\bibliographystyle{plain}
\bibliography{/home/thomas/Dropbox/tp}
\end{document}